\documentclass[
  prx,
  twocolumn,
  superscriptaddress,
  longbibliography
]{revtex4-2}

\usepackage[english]{babel}

\usepackage[a4paper,top=2cm,bottom=2cm,left=2cm,right=2cm,marginparwidth=1.75cm]{geometry}

\usepackage{amsmath}

\usepackage{comment}
\usepackage{float} % <-- Important

\usepackage{mathtools}
\usepackage{amssymb}
\usepackage{graphicx}

\usepackage{placeins}

\usepackage[ruled,linesnumbered]{algorithm2e}
\usepackage{algpseudocode}
\usepackage{braket}
\usepackage{amsthm}
\newtheorem{theorem}{Theorem}

\usepackage{aliascnt}
\usepackage[colorlinks=true, allcolors=blue]{hyperref}

\newaliascnt{lemma}{theorem}
\newtheorem{lemma}[lemma]{Lemma}
\aliascntresetthe{lemma}

\newaliascnt{definition}{theorem}
\newtheorem{definition}[definition]{Definition}
\aliascntresetthe{definition}

\newaliascnt{conjecture}{theorem}
\newtheorem{conjecture}[conjecture]{Conjecture}
\aliascntresetthe{conjecture}

\newcommand*{\order}[1]{\mathcal{O}\left(#1\right)}
\newcommand*{\abs}[1]{\left|#1\right|}
\newcommand*{\hil}{\mathcal{H}}
\newcommand*{\complex}{\mathbb{C}}

\begin{document}

\title{Constraint-oriented biased quantum search for general constrained combinatorial optimization problems}
\author{Sören Wilkening}
\affiliation{Institut f\"ur Theoretische Physik, Leibniz Universit\"at Hannover, Germany}

\begin{abstract}
We present a quantum algorithmic routine that extends the realm of Grover-based heuristics for tackling combinatorial optimization problems with arbitrary efficiently computable objective and constraint functions. 
Building on previously developed quantum methods that were primarily restricted to linear constraints, we generalize the approach to encompass a broader class of problems in discrete domains.
To evaluate the potential of our algorithm, we assume the existence of sufficiently advanced logical quantum hardware.
With this assumption, we demonstrate that our method has the potential to outperform state-of-the-art classical solvers and heuristics in terms of both runtime scaling and solution quality.
The same may be true for more realistic implementations, as the logical quantum algorithm can achieve runtime savings of up to $10^2-10^3$.
\end{abstract}
\maketitle

\section{Introduction}
Grover's search algorithm \cite{Durr1996AQuantumAlgorithmForFindingTheMinimum, Grover1996AFastQuantumMechanicalAlgorithmForDatabaseSearch} is one of the most promising contenders for achieving a quantum advantage over classical algorithms for unstructured database search. 
The concept also applies to combinatorial optimization problems \cite{Nemhauser1988IntegerAndCombinatorialOptimization}, typically NP-hard problems \cite{Karp1975NPHardness}.
There has been some debate about whether Grover's algorithm can deliver a practical quantum advantage or not \cite{Stoudenmire_2024, Aaronson2023blog}.
As shown by Wilkening et al. \cite{wilkening2024quantumalgorithmsolving01, wilkening2025quantumsearchmethodquadratic, ODSQ1}, Grover's quantum search algorithm, when applied to combinatorial optimization problems, can be tailored to specific instances using advanced state-preparation techniques, potentially outperforming general-purpose state-of-the-art classical algorithms, such as Gurobi \cite{gurobi}.
The tackled problems comprised linear and quadratic objective functions and general linear constraints.
The main idea of the approach is to insert simple checks while computing the constraint violation to reduce branching during state preparation for any produced state.
In this work, we generalize the framework introduced in \cite{ODSQ1}, namely the Constraint-oriented biased quantum search (CBQS), to tackle a broader range of optimization problems, potentially any problem from the optimization class NPO \cite{Leeuwen2012TheoryOfComputingSystems}.
This implies that we only tackle problems for which the validity of a given assignment can be computed efficiently (in polynomial time).
We provide a formulation for single-objective optimization problems, but the same ideas can also be applied to multi-objective optimization problems \cite{Deb2014}.
Numerical experiments demonstrate the framework's potential as a competitor against classical and quantum general-purpose state-of-the-art solvers. 
The classical methods chosen are Gurobi and Hexaly \cite{hexaly}, as well as a simulated annealing tool, Simanneal \cite{simanneal}.
Additionally, a simplified comparison is made against quantum Branch-and-Bound (qBnB) \cite{Montanaro_2020,chakrabarti2022universalquantumspeedupbranchandbound,Ambainis_2017}.

We are aware of QAOA implementations addressing pseudo-boolean optimization problems \cite{yue2023localglobaldistributedquantum}, but we disregard them for three reasons. 
Firstly, when simulating a general quantum algorithm, the number of simulateable qubits \cite{Lykov_2023} for high depth circuits is a bottleneck, only allowing for the investigation of tiny instances with $\leq30$ variables (depends on how many ancillary qubits are required). 
Even though there are more advanced simulation techniques based on tensor network methods \cite{dubey2025simulatingquantumcircuitstree} that can simulate more qubits, these have only been tested on extremely shallow circuits and may therefore not be helpful for the considered benchmark instances in this paper.
This would still be an interesting question to investigate, but it is beyond the scope of his work.
We can use specialised implementations to investigate larger instances (up to 3000 variables) of our proposed quantum algorithm.
Secondly, our state preparation unitary can be used to generate an initial state for QAOA, as investigated for the 0-1 Knapsack problem \cite{christiansen2024quantumtreegeneratorimproves}.
The study demonstrated that implementing QAOA with custom state preparation, similar to ours, is superior to plain QAOA algorithms.
We expect the same outcome for the implementation provided in this paper.
For the numerical analysis, a quadratic objective function and quadratic constraints are chosen.
Thirdly, we aim to tackle more general problems than those represented by pseudo-Boolean formulas with an efficient polynomial representation.
For example, product terms (efficiently implementable via multiplication) pose a significant challenge for QAOA, as they can't be easily mapped to an efficient Hamiltonian representation.

\section{Previous work}
An early method to incorporate the structure of combinatorial optimization problems into Grover's search was \textit{Nested quantum search} \cite{Cerf2000Nested!uantumSearchAndStructuredProblems}.
The algorithm aims to evaluate constraint violations on a subset of variables to filter out constraint-violating partial assignments via Grover search.
This is done on multiple disjoint subsets of variables.
The resulting quantum state is then used as an initial state for an outer Grover search to solve the optimization problem.

A recent competitive approach, a Grover-based quantum heuristic, was made by Wilkening et al. \cite{wilkening2024quantumalgorithmsolving01, wilkening2025quantumsearchmethodquadratic}, which tackles the 0-1, quadratic, and multidimensional knapsack problems.
The developed quantum routine aimed to build the initial quantum state iteratively, eliminating every non-feasible solution at each iteration.
Therefore, the state-preparation oracle $\mathcal{G}$ generates a state only containing feasible solutions to the optimization problem. 
The primary limitation of this approach is that it only allows for benign constraint functions, thereby limiting its applicability.
This problem was addressed in \cite{ODSQ1}, where a method was proposed to tackle general linear-constraint integer problems. 
But this method cannot avoid generating non-feasible solutions.
Nevertheless, all the methods were developed for linear constraint functions.
The basic framework of the proposed algorithm uses Grover's algorithm to formulate a quantum heuristic, incorporating problem-specific initial-state preparation and a non-exponential cutoff for oracle application.
The search routines are described by \autoref{alg:QSearch} and \ref{alg:QMaxSearch}, where $S_0 = 2\ket{0}\bra{0} - 1$ and $S_F = (-1)^{F(x)}$.

\begin{algorithm}[t]
\caption{\textbf{QSearch}($\mathcal{G}$, $F$, $M$)}\label{alg:QSearch}
%\begin{algorithmic}[1]
%\Function{\textbf{QSearch}}{$\mathcal{G}$, $F$, $M$}
    Set $l = 0$, $m_{tot} = 0$, and let $d$ be any constant such that $1 < d < 2$\;
    Increase $l$ by $1$ and set $m = \lceil d^{l}\rceil$
    Apply $\mathcal{G}$ to the initial state $\ket{0} \ket{c_1}\dots\ket{c_m} \ket{0}$\;
    Choose $j \in [1, m]$ uniformly at random
    $m_{tot} \gets m_{tot} + 2j + 1$\;
    Apply $\left(\mathcal{G}S_0\mathcal{G}^\dagger S_F\right)^{j}$ to $\ket{s} = \mathcal{G} \ket{0} \ket{c_1}\dots\ket{c_m} \ket{0}$\;
    Measure all registers $\rightarrow \ket{x} \ket{\tilde{c}_1} \dots \ket{\tilde{c}_m} \ket{P_{x}}$\;
    \eIf{$F(P_{x}) = 1$ \textbf{and} all $\tilde{c}_k \geq 0$ \textbf{or} $m_{tot} \geq M$}{
        \textbf{return} $x$, $P_{x}$\;
    }{
        go to step 2\;
    }
%\EndFunction
%\end{algorithmic}
\end{algorithm}

Although this approach shares some similarities with Nested Quantum Search (NQS), Wilkening et al. \cite{wilkening2024quantumalgorithmsolving01} have demonstrated empirically that this method outperforms NQS in terms of performance.
A reason for this is that filtering out constraint-violating partial assignments using Grover's algorithm may require an exponential number of oracle applications, whereas iterative constraint checking using circuit operations only requires a depth proportional to the number of bits needed to store the constraint information and is therefore unlikely to be exponentially expensive.

Furthermore, another approach is Quantum Branch-and-Bound (qBnB) \cite{Montanaro_2020}, which achieves a Grover speedup over the classical Branch-and-Bound methodology by exploring a Tree of non-constraint-violating nodes with potential objective values better than a given threshold.
The caveat of this approach compared to Nested Quantum Search and our proposed method is that qBnB comprises potentially very expensive oracle calls \cite{huang2021branchboundmixedinteger, nannicini2022fastquantumsubroutinessimplex,ammann2023realisticruntimeanalysisquantum}.
A more detailed description is provided in \autoref{sec:benchmarking}.

\begin{algorithm}[t]
\caption{\textbf{QMaxSearch}($\mathcal{G}$, $M$, $p$, $w$, $c$)}\label{alg:QMaxSearch}
%\begin{algorithmic}[1]
%\Function{\textbf{QMaxSearch}}{$\mathcal{G}$, $M$, $p$, $w$, $c$}
set $T=0$, $x=0$\;
\While{True}{
    $x$, $P_{x}$ = \textbf{QSearch}($\mathcal{G}$, $F(x) = x > T$, $M$)\;
\eIf{$P_{x} > T$}{
    set $T = P_{x}$ and $x_{out} = x$\;
}{
    \textbf{return} $x_{out}$, $T$\;
}
}
%\EndFunction
%\end{algorithmic}
\end{algorithm}

\section{Method}

Without loss of generality, we assume an optimization problem with $n$ only binary variables $x_j\in\{0,1\}$ $\forall$ $1\leq j\leq n$. 
Problems with $n$ bounded integer variables $x_j\in\{0, b\}$ could be reduced to a binary problem using $n\log_2(b)$ binary variables.
Furthermore, we consider our problem to contain only integer coefficients within the objective and constraint functions, as scaling up fractional values to integers would yield the same algorithm. 
As we store integers in their binary representation, this ensures the efficiency of the quantum routines.

Grover's algorithm and amplitude amplification \cite{Brassard_2002} are typically defined given an arbitrary function $f\rightarrow \{0,1\}^n\rightarrow \{0,1\}$; for completeness, we include a statement on the potential implementation of arbitrary objective functions.

\begin{theorem}
	\label{theorem:general_objective}
	For any optimization problem with an efficiently computable objective function, there exists a quantum circuit that efficiently computes the respective objective value of any solution contained in a quantum state $\ket{\psi}$.
\end{theorem}
\begin{proof}
	As the computation of the objective function does not, constructed in \cite{wilkening2024quantumalgorithmsolving01, ODSQ1, wilkening2025quantumsearchmethodquadratic} and this work, affect the set of states contained in $\ket{\psi}$, the state preparation consists of the two unitaries $U^{obj}$ and $U^{con}$, s.t. $\ket{\psi} = U^{con}\ket{\psi_0}$.
	Furthermore, as shown in \cite{Bernstein1997QuantumComplexityTheory}, any reversible efficient Turing machine can be simulated efficiently by a quantum Turing machine. For every efficiently computable objective, there exists a quantum circuit with polynomial depth that computes the objective value of every solution within $\ket{\psi}$.
\end{proof}

To develop a more general framework for the constraint-respecting part of the state preparation routine, we will begin by defining explicit quantum operations that tackle a wide variety of problems.
As the objective and constraint functions map an $n$-dimensional binary input onto a natural (real) number, they are called \textit{pseudo-boolean functions} \cite{Boros2002PseudoBooleanOptimizations}.

\begin{lemma}
    Let $n$ be an integer and $V=\{1, \dots, n\}$. 
    Every pseudo-boolean function $f:\{0,1\}^n\rightarrow \mathbb{R}$ can be represented by a \textit{multi-linear polynomial} of the form
    \begin{align}
    \label{eq:mlpf}
        f(x_1,\dots,x_n) = \sum_{S\subseteq V}p_S \prod_{j\in S}x_j.
    \end{align}
\end{lemma}
\begin{proof}
    See proof of theorem 2 in \cite{Boros2002PseudoBooleanOptimizations}.
\end{proof}

\autoref{eq:mlpf} is very useful for an algorithm description to tackle arbitrary objective and constraint functions, but it might lead to an inefficient formulation of specific optimization problems.
An example of this would be the product Knapsack problem \cite{pferschy2021productknapsackproblemapproximation}.
Decomposing its objective function into a multi-linear polynomial form would result in a function with $\mathcal{O}(2^n)$ terms, making it inefficient for evaluating any quantum or classical algorithm.
Therefore, w.l.o.g., we extend \autoref{eq:mlpf} by a term incorporating an efficient representation of product terms.
Assuming $f$ contains $Z$ product functions we get
\begin{widetext}
\begin{align}
\label{eq:general_function}
    f(x_1,\dots,x_n) = \sum_{S\subseteq V}p_S \prod_{j\in S}x_j 
    + \sum_{i = 1}^Z \alpha_i \prod_{S\subseteq V} \mathrm{pow}\left(p'_{i,S}, \prod_{j\in S'}x_j\right).
\end{align}
\end{widetext}
Here $\alpha_i \in \{-1, 0, 1\}$.
The constant could be merged into $p'_{i,S}$, but is used to indicate the number of product terms present in the function.
The added term is a compactification of a specific part of the multi-linear polynomial form, resulting in a more efficiently implementable expression.

Here we assume $p$ and $p'$ to be a tensors with 
\begin{align}
    &q=\left|\{S:S\subseteq V, p_S\neq 0\}\right| \text{ and } \\
    &q' = \left|\{(i,S):S\subseteq V, p'_{i, S}\neq 0, 1\leq i \leq Z\}\right|
\end{align}
non zero elements of rank $r=\max\{\left|S\right|:S\subseteq V, p_S\neq 0\}$ and $r'=\max\{\left|S\right| + 1:S\subseteq V, p'_{i,S}\neq 0, 1\leq i\leq Z\}$ respectively, consisting of only constant integer values.
More general functions $p(S)$ could be considered, but constant coefficients are sufficient for this paper.

We can also use the formulation of pseudo-boolean functions for arbitrary constraint functions, leading to the most general formulation
\begin{align}
\label{eq:MostGeneralConstraint}
\begin{aligned}
    &\sum_{S\subseteq V} w_{S,k} \prod_{j\in S} x_j + \\ 
    &\sum_{i = 1}^Z \alpha_{i,k} \prod_{S\subseteq V} \mathrm{pow}\left(w'_{i,S,k}, \prod_{j\in S'}x_j\right) \leq c_k
\end{aligned}
\end{align}

\subsection{The sum constraint}
The sum and product expressions require different approaches, and for simplicity, we develop quantum routines for both cases separately.
For now, we therefore focus on constraints of the form 
\begin{align}
    \label{eq:sum_constraint}
    \sum_{S\subseteq V} w_{S,k} \prod_{j\in S} x_j \leq c_k.
\end{align}
As done in \cite{ODSQ1}, using a shift $P_k = \sum_{S\subseteq V:w_{S,k} < 0} \abs{w_{S,k}}$ this can be reformulated into 
\begin{align}
\label{eq:constraints_adusted}
\begin{aligned}
    \sum_{S\subseteq V} &\left|w_{S,k}\right| y_{S,k} \leq c_k + P_k \text{, where } \\
    &y_{S,k} = 
    \begin{cases}
        \prod_{j\in S} x_j &\text{ if } w_{S,k} \geq 0\\
        1- \prod_{j\in S} x_j &\text { else }
    \end{cases}
\end{aligned}
\end{align}
As the number of variables increases, so extra constraints are introduced stating the relationship of these variables, f.e. that $y_{S,k} = y_{S,k'}$ when $sign(w_{S,k}) = sign(w_{S,k'})$.
Nevertheless, in the remainder of the work, we are not interested in these additional constraints, as they are redundant when assigning the original variables $x_j$.
We only require \autoref{eq:constraints_adusted} as a description for the quantum routine developed in this work.
Before we build our quantum algorithm, we require some definitions.
\begin{definition}
    (Closed product) We call the product $\prod_{j\in S}x_j$ closed if every variable $x_j$ is assigned or at least one of the $x_j$ is assigned to 0.
\end{definition}
\begin{definition}
    We call a variable $y_{S,k}$ closed if its product $\prod_{j\in S}x_j$ is closed. 
    More formally, $y_{S,k}$ is closed given a partial assignment with $t$ assigned variables, if $t \geq \max(S)$.
\end{definition}

\begin{definition}
\label{def:potential_satisfyable}
    We call a partial assignment $(x_1,\dots,x_t)$ with $t$ assigned variables \textit{potentially satisfiable}, if every constraint of \autoref{eq:MostGeneralConstraint} individually could be satisfied.
\end{definition}
As in previous work \cite{wilkening2024quantumalgorithmsolving01, wilkening2025quantumsearchmethodquadratic, ODSQ1} to verify the potential satisfiability of a given partial assignment, we disregard the contribution of non-closed variables by assigning them to $y_{S,k} = 0$. 
As mentioned, we disregard the variables' relationship constraints, as this is just an implied and not a fixed assignment.
Notably, a potentially satisfiable assignment does not necessarily lead to a satisfying solution.

\begin{lemma}
\label{lemma:oracle}
    Let $\kappa$, $n$ and $t$ be integers, s.t. $t\leq n$. 
    Further, let $w$ be a tensor with $\order{poly(n)}$ non-zero elements.
    Given a quantum state containing a potentially satisfiable assignment $(x_1,\dots,x_t)$.
    There exists a quantum oracle that, if at least one of the two successors $(x_1,\dots,x_t,0)$ and $(x_1,\dots,x_t,1)$ is potentially satisfiable, generates only the potentially satisfiable assignments, and otherwise does nothing.
    This oracle is implementable with depth of $\order{poly(n)m(\log_{2}(n) + \log_{2}(\kappa))}$ using $\order{n + m\kappa}$ qubits.
\end{lemma}
\begin{proof}
    The remaining potential of a given assignment is defined as
    \begin{align}
        \widetilde{c}_k = c_k + P_k - \sum_{S\subseteq V: y_{S,k}\text{ is closed}} \abs{w_{S,k}} y_{S,k}.
    \end{align}
    Assuming a quantum state $\ket{\psi}\in \mathbb{C}^{n}\otimes \bigotimes_{k=1}^m\mathbb{C}^{\kappa}$ containing the potentially satisfiable assignment $(x_1,\dots,x_t)$. 
    W.l.o.g., we consider the state to be
    \begin{align}
        \ket{\psi} = \ket{x_1}\cdots\ket{x_t}\ket{0}^{n-t}\ket{\widetilde{c}_1}\cdots\ket{\widetilde{c}_m}.
    \end{align}
    The proof works the same way when we assume the state is in a superposition. 
    As all the terms with closed variables $y_{S,k}$ are already considered within $\widetilde{c}_k$, we need to evaluate the potential cost of assigning $x_{t+1}$.
    As the constraints act differently for positive and negative coefficients, we evaluate values
    \begin{align}
        \label{eq:Q_plus}
        Q^+_{k,t+1} &= \sum_{
            \substack{S\subseteq V\\\max(S) =  t+1,w_{S,k}\geq0}
        } \abs{w_{S,k}} \prod_{j\in S/\{t+1\}}x_j,\\
        \label{eq:Q_minus}
        Q^-_{k,t+1} &= \sum_{
            S\subseteq V: w_{S,k}<0,t+1\in S
        } \abs{w_{S,k}} \prod_{j<t+1\in S}x_j
    \end{align}
    The sum considers only terms in which $S$ contains $t+1$.
    The sum of terms with positive coefficients only requires all the terms where $t+1$ is the greatest item index, while for the sum of terms with negative coefficients, all terms containing $x_{t+1}$ have to be considered.
    Based on \autoref{eq:Q_plus} and \ref{eq:Q_minus} we can formulate four boolean expressions
    \begin{align}
    \label{eq:boolean}
        b^\pm_{t+1} = \bigwedge_{k=1}^m \left(Q^\pm_{k,t+1}\leq \widetilde{c}_k\right).
    \end{align}
    $b^\pm_{t+1}$ determine wether $x_{t+1}$ can be assigned to $0$ or (and) $1$ and therefore identifies potentially satisfiable assignments.
    As we aim to generate at least all the satisfying assignments to the original optimization problem, we define an oracle with four possible actions:
    \begin{enumerate}
        \item if $b^+_{t+1} \land b^-_{t+1}$: apply $R_y$ to $\ket{0}_{t+1}$
        \item if $b^+_{t+1} \land \lnot b^-_{t+1}$: apply $X$ to $\ket{0}_{t+1}$
        \item if $\lnot b^+_{t+1} \land b^-_{t+1}$: do nothing
        \item if $\lnot b^+_{t+1} \land \lnot b^-_{t+1}$: do nothing
    \end{enumerate}
    After executing the potential assignments, the remaining capacities have to be adjusted properly, resulting in the desired unitary
    \begin{align}
    \begin{aligned}
        U_{t+1} = &\prod_{
            k=1
        }^m \Lambda^{x_{t+1}}(SUB(C_k,Q^+_{k,t+1})) \times\\
        &\prod_{
            k=1
        }^m \Lambda^{\lnot x_{t+1}}(SUB(C_k,Q^-_{k,t+1})) \times \\
        &\Lambda^{b^+_{t+1}\land b^-_{t+1}} (R_{y,t+1}(\theta)) \Lambda^{b^+_{t+1}\land \lnot b^-_{t+1}} (X)
    \end{aligned}
    \end{align}
    Afterwards $Q^{\pm}_{{k,t+1}}$ have to be uncomputed.
    In the worst case every positive/negative coefficient $w_{S,k}$ has to be considered, which results in $\order{\sum_{k=1}^{m}q_{k}} = \order{poly(n)m}$ adding operations, that are controlled by $\order{r_{k}}$ qubits. 
    The depth of a single controlled adding operation is given by $\order{\log_{2}(r_{k}) + \log_{2}(\kappa)}$.
    The integer comparison to compute $b^{\pm}_{{t+1}}$ just adds a small overhead of $\order{m\log_{2}(\kappa)}$.
\end{proof}

An essential fact of the algorithm presented in \autoref{lemma:oracle} is that whenever the oracle has a potentially satisfying state as an input, where the two successors are not satisfiable, the initial state cannot be undone and therefore remains as an unfeasible state. 
Therefore, the quantum procedure, especially when applied repeatedly to assign multiple variables, cannot guarantee a quantum state with only feasible solutions to \autoref{eq:sum_constraint}.
Still, as observed in practice, the quantum state will generate fewer states than the uniform superposition covering every possibility.

\begin{theorem}
\label{theorem:constraint}
    Let $\kappa$, $n$ be integers and $w$ be a tensor with $\order{poly(n)}$ non-zero elements.
    There exists a quantum algorithm that iteratively generates at least all feasible solutions to \autoref{eq:sum_constraint}, but at every iteration, it only produces infeasibilities if the two successors of a potentially satisfiable solution are not satisfiable.
    The algorithm runs in time  $\order{poly(n)m(\log_{2}(n) + \log_{2}(\kappa))}$ using $\order{n + m\kappa}$ qubits.
\end{theorem}
\begin{proof}
	Let $\kappa = \order{\max_{k}\{\log_{2}(c_{k} + P_{k})\}}$ be the qubit number to store all integers accurately. 
	For the algorithm we define a Hilbert space $\hil=\mathbb{C}^{n}\otimes \bigotimes_{k=1}^m\mathbb{C}^{\kappa}$, which we initialize with the quantum state
	\begin{align}
		\ket{\psi_{0}} = \ket{0}^{n} \ket{c_{1} + P_{1}}\cdots\ket{c_{m} + P_{m}}.
	\end{align}
	We can see $\ket{\psi}$ as the most satisfactory partial assignment because it is the least restricted one.
	As $c_k + P_k - \sum_{S\subseteq V: y_{S,k}\text{ is closed}} \abs{w_{S,k}} y_{S,k}$ determines the remaining potential to satisfy a given partial assignments, applying the oracle defined in \autoref{lemma:oracle} gives the quantum state
	\begin{align}
		\ket{\psi} = U_{n}\dots U_{1}\ket{\psi_{0}}.
	\end{align}
	$\ket{\psi}$ will contain at least all the satisfying states since $U_{t}$, as shown in \autoref{lemma:oracle} will at every level generate the set of potentially satisfiable partial assignments.
	The cost of the algorithm is determined by applying $U_{t}$ $n$ times, resulting in the runtime $\order{poly(n)m(\log_{2}(n) + \log_{2}(\kappa))}$.
\end{proof}

This section already provides a framework to tackle arbitrary constraint functions with an efficient polynomial expression. 
However, as discussed, product expressions do not have this efficient polynomial form; nevertheless, we are still able to describe a poly-depth quantum circuit for state preparation.

\subsection{The product constraint} 
The strategy to cut non-satisfiable parts from the superposition over all states is to evaluate the smallest possible evaluated constraint function to check whether there could be at least one assignment satisfying the constraints.
As this is again an NP-hard problem, providing a simple lower bound suffices. 
The better the lower bound, the higher the amplitudes of the generated solution.
It becomes, therefore, a meta-optimization problem weighing the cost of tighter lower bounds against the number of Grover iterations required to find a (or the optimal) solution.

When we are concerned with the product terms of \autoref{eq:MostGeneralConstraint}, we look at constraints of the form
\begin{align}
\label{eq:product_constraint}
    \sum_{i = 1}^Z \alpha_{i,k}\prod_{S\subseteq V} \mathrm{pow}\left(w'_{i,S,k}, \prod_{j\in S}x_j\right) \leq c_k.
\end{align}

In contrast to the sum expressions, some dependencies of already fixed assignments must be considered to compute the smallest possible value of a constraint function. 
If, for example, the previous assignments evaluated the constraint to a negative value, the remaining terms in the product should be increased to their maximum possible value.
To define certain rules, we need some clarification.
\begin{definition}
    We call the value $w'_{i,S,k}\neq 0$ of the term $\mathrm{pow}\left(w'_{i,S,k}, \prod_{j\in S}x_j \right)$, given a partial assignment $x$, \textit{remaining}, if none of the variables $x_j$ are assigned $0$ or not all variables $x_j$ for all $j\in S$ are assigned.
\end{definition}

\begin{algorithm}[t]
    \caption{LB($i, k, x_{par}$)}
    \label{alg:lower_bound}
    %\begin{algorithmic}
      %  \Function{LB}{$i, k, x_{par}$}
    \If{all remaining $w'_{i,S,k}>0$ \textbf{and} $g_{i,k}(x_{par})<0$}{
        \Return product of all remaining $w'_{i,S,k}$\;
    }
    \If{all remaining $w'_{i,S,k}>0$ \textbf{and} $g_{i,k}(x_{par})>0$}{
        \Return 1\;
    }
    \If{some remaining $w'_{i,S,k}<0$ \textbf{and} $g_{i,k}(x_{par})<0$ \textbf{and} \#(remaining $w'_{i,S,k}<0$) is even}{
        \Return product of all remaining $w'_{i,S,k}$\;
    }
    \If{some remaining $w'_{i,S,k}<0$ \textbf{and} $g_{i,k}(x_{par})<0$ \textbf{and} \#(remaining $w'_{i,S,k}<0$) is odd}{
        \Return product of all remaining $w'_{i,S,k}$  except largest $w'_{i,S,k}<0$\;
    }
    \If{some remaining $w'_{i,S,k}<0$ \textbf{and} $g_{i,k}(x_{par})>0$ \textbf{and} \#(remaining $w'_{i,S,k}<0$) is even}{
        \Return product of all remaining $w'_{i,S,k}$\;
    }
    \If{some remaining $w'_{i,S,k}<0$ \textbf{and} $g_{i,k}(x_{par})>0$ \textbf{and} \#(remaining $w'_{i,S,k}<0$) is odd}{
        \Return product of all remaining $w'_{i,S,k}$  except largest $w'_{i,S,k}<0$\;
	}

\end{algorithm}

Here two implementations are possible for \autoref{alg:lower_bound}.
Either all the classically known remaining terms are considered, taking only into account wether ofr not all variables are assigned, or incorporating the knowledge of the assignments into the quantum routine.
Method one results in a quantum circuit of lower depth, while it will provide worse lower bounds resulting in a greater number of Grover iterations.
This tradeoff would be interesting to investigate, yet it beyond the scope of this paper.

We also have to define a function that evaluates the product functions based on a given partial assignment:
\begin{align}
\begin{aligned}
    &g_{i,k}(x_{{par}}) = \\ 
    &\prod_{S\subseteq V:w'_{i,S,k} \text{ not remaining}} \mathrm{pow}\left( w'_{i,S,k}, \prod_{j\in S}x_{par, j} \right)
\end{aligned}
\end{align}

We can build an algorithm to compute a simple lower bound on the minimal achievable value of \autoref{eq:product_constraint} as described in \autoref{alg:lower_bound}.

We will apply the exact definition of partial satisfiability given by \autoref{def:potential_satisfyable} with respect to the lower bound heuristic given by \autoref{alg:lower_bound}.

\begin{definition}
	\label{def:product_constraint}
	We call a partial assignment $x_{par}$ concerning product constraints potentially satisfyable, if for some lower bound algorithm $LB$, we are given
	\begin{align}
		\label{eq:prod_lower_bound}
		\sum_{i=1}^{Z} \alpha_{i,k} g_{i,k}(x_{par}) LB(i,k, x_{par}) \leq c_{k}
	\end{align}
\end{definition}

Using this definition, we can define an oracle with the same purpose as defined by \autoref{lemma:oracle}.
\begin{lemma}
\label{lemma:prod_constr}
	Let $\kappa$, $n$ and $t$ be integers, s.t. $t\leq n$ and $V=\{1,\dots,n\}$. 
	Further, let $w'$ be tensor with $\order{poly(n)}$ non-zero elements $\in\mathbb{N}/\{0\}$.
	Given a quantum state containing a potentially satisfiable assignment $(x_1,\dots,x_t)$.
	There exists a quantum oracle that, if at least one of the two successors $(x_1,\dots,x_t,0)$ and $(x_1,\dots,x_t,1)$ is potentially satisfiable, generates all and only the potentially satisfiable assignments, and otherwise does nothing.
	This oracle is implementable with a depth of
	\begin{align}
		\order{poly(n)mZ(\log_{2}(n) + \log^2_2(\kappa))}
	\end{align}
	using $\order{n + poly(n)mZ\kappa}$ qubits.
\end{lemma}
\begin{proof}
	We assume that all integers are stored in sufficiently large quantum registers with at most $\kappa$ qubits.
	For any partial assignment, we also aim to store the fixed value of all the $Z$ terms of all the $m$ constraints individually in quantum registers, resulting in the quantum state
	\begin{align}
		\ket{\psi} = \ket{x_{1}}\dots\ket{x_{t}}\dots\ket{0}\ket{g_{1,1}(x_{par})}\dots{\ket{g_{Z,m}(x_{par})}}\ket{0}
	\end{align}
	Here, we assume that all $\alpha_{i,k}\neq 0$, since the value does not need to be stored.
	To check whether any successor of $x_{par}$ is satisfiable, we have to assign $x_{t+1}$ to 0 and one and evaluate all the constraints.
	For that we append $\ket{x'_{par}}\gets\ket{x_{par}}\ket{x_{t+1}}$.
	Using in the worst case $\order{poly(n)}$ controlled multiplications we get
    %\begin{widetext}
	\begin{align}
    \begin{aligned}
		\ket{\psi'} &= U^{1,x_{t+1}}\ket{\psi} = \\
         &\ket{x_{1}}\dots\ket{x_{t}}\dots\ket{0}\ket{g_{1,1}(x'_{par})}\dots{\ket{g_{Z,m}(x'_{par})}}\ket{0},
    \end{aligned}
	\end{align}
    %\end{widetext}
	resulting in a circuit depth of $\order{poly(n)\left(\log_{2}(n) + \log^2_2(\kappa)\right)}$ \cite{Junhong2023QuantumCircuitDesignForIntegerMultiplicationBasedOnSchönhageStrassenAlgorithm}.
    As multiplication on a quantum computer is only performed out-of-place, this computation requires $\order{poly(n)nZ\kappa}$ ancilla qubits.
	Storing $x_{t+1}$ in the quantum register is not necessary yet, as it might not be feasible.
	Computing the lower bound as in \autoref{alg:lower_bound}, requires that for every term that contains any of the assigned variables, a check has to be made, if any of them are assigned to $0$.
	In the worst case this has to be done with a Toffoli on $n$ qubits for $\order{poly(n)}$ terms for every function for every constraint.
	Further the quantum implementation of \autoref{alg:lower_bound} requires the sign of the $ mz$ stored terms as quantum input. 
	To perform
    %\begin{widetext}
	\begin{align}
	\begin{aligned}
		\ket{\psi^{LB}} &= U^{LB,x_{t+1}}\ket{\psi'} = \\
        &\ket{x_{1}}\dots\ket{x_{t}}\dots\ket{0}\ket{g_{1,1}(x'_{par})LB(1,1,x'_{par})}\otimes\\
        &\dots\otimes\ket{g_{Z,m}(x'_{par})LB(Z,m,x_{'par})}\ket{0}
	\end{aligned}
	\end{align}
    %\end{widetext}
	we therefore demand a circuit depth of $\order{poly(n)\left(\log_{2}(n) + \log^2_2(\kappa)\right)}$.
	Using a unitary $U^{Cmp,x_{t+1}}$ with a depth of 
	\begin{align}
		\order{m(Z+1)\log_{2}(\kappa) + \log_{2}(m)}
	\end{align}
	we evaluate the expression \autoref{eq:prod_lower_bound} for every constraint and store the boolean of the feasibility of all constraints in $b^{x_{t+1}}$.
	To determine the potential assignments $x_{t+1} = 0$ and $x_{t+1} = 1$, we compute $b^{0}$ and $b^{1}$ using:
    %\begin{widetext}
	\begin{align}
    \begin{aligned}
		U = &\left(U^{Cmp,1}U^{LB,1}U^{1,1}\right)^{\dagger}\Lambda^{b_{1}}(X^{an})
        U^{Cmp,1}U^{LB,1}U^{1,1}\times\\
        &\left(U^{Cmp,0}U^{LB,0}\right)^{\dagger}\Lambda^{b_{0}}(X^{an})U^{Cmp,0}U^{LB,0}
    \end{aligned}
	\end{align}
    %\end{widetext}
	With the returned booleans, we define the assignment routine for $x_{t+1}$ in the same way as done in \autoref{lemma:oracle}.
	In the last step $b_{0}$ and $b_{1}$ have to be uncomputed.
\end{proof}

The cost of the classical computation of \autoref{alg:lower_bound} in \autoref{lemma:prod_constr} is not accounted for, as it only has to be done for the first state preparation circuit and not every time the oracle has to be called again. 

\begin{theorem}
\label{theorem:product_constraint} 
    Let $\kappa$, $n$ be integers and $w$ be a tensor with $\order{poly(n)}$ non-zero elements.
    There exists a quantum algorithm that iteratively generates at least all feasible solutions to \autoref{eq:product_constraint}, but at every iteration, it only produces infeasibilities if the two successors of a potentially satisfiable solution are not satisfiable.
    The algorithm runs in time  $\order{poly(n)mZ(\log_{2}(n) + \log^2_2(\kappa))}$ using $\order{n + poly(n)mZ\kappa}$ qubits.
\end{theorem}
\begin{proof}
	To apply the unitary from \autoref{lemma:prod_constr}, we need to define a Hilbert space
	\begin{align}
		\hil = \complex^{2^{n}} \otimes \bigotimes_{k=1}^{m} \left(\bigotimes_{i=1}^{Z+1} \complex^{2^{\kappa}}\right) \otimes \complex^{anc}
	\end{align}
	to store the integer assignments, the $m$ total constraint expressions, and the $mZ$ stored fixed subexpressions to compute the constraint, as well as an ancilla register for simplex computations.
    The ancilla registers have to be large enough to store the intermediate results of the out-of-place quantum multiplications, requiring $\order{mZpoly(n)}$ qubits.
	The initial state has to be set to
	\begin{align}
		\ket{\psi_{0}} = \ket{0}^{n} \otimes \bigotimes_{k=1}^{m} \left(\bigotimes_{i=1}^{Z} \ket{1}\ket{0}\right) \otimes \ket{0}.
	\end{align}
	The registers for the subexpressions have to be assigned to $1$ to evaluate the correct products for any assignment.
	Applying the algorithm described in \autoref{lemma:prod_constr} iteratively for every variable of the optimization problem, will at every state generate only potentially satisfiable solutions based on \autoref{def:product_constraint} to a partial solution, if at least one of its successors is potentially satisfiable.
\end{proof}

It is critical to remember that the algorithm can generally not provide a quantum state only containing truly satisfying assignments, as the definition of potential satisfiability of partial assignments is weaker than the satisfiability of whole assignments.
The routines defined in \autoref{theorem:general_objective}, \ref{theorem:constraint}, and \ref{theorem:product_constraint} combined describe the quantum routine that tailors Grover's algorithm to a wide range of optimization problems, given efficient representations in the form of \autoref{eq:MostGeneralConstraint}.
Even though the described methods don't capture every potential operation for constraint functions, even when tackling problems like $1/g_1(x) + 1/g_2(x)\leq c$, the same methodology applies.
Further, as done in \autoref{theorem:general_objective}, we believe that there has to be an efficient state-preparation computing only potentially satisfiable assignments for any possible efficiently computable constraint function, given some definition of potential satisfiability.

\begin{conjecture}
\label{conjecture:general_problems}
	For every problem in NPO \cite{Leeuwen2012TheoryOfComputingSystems}, there exists an iterative quantum algorithm that generates only potentially feasible solutions, if any successor of a partial solution is potentially satisfiable based on a non-trivial constraint-check.
\end{conjecture}

\begin{algorithm}[t]
\caption{\textbf{QSearchViolation}($\mathcal{G}$, $F$, $M$)}\label{alg:QSearch2}
%\begin{algorithmic}[1]
%\Function{\textbf{QSearchViolation}}{$\mathcal{G}$, $F$, $M$}
Set $l = 0$, $m_{tot} = 0$, and let $1 < d < 2$ be constant\;
Increase $l$ by $1$ and set $m = \lceil d^{l}\rceil$\;
Apply $\mathcal{G}$ to the initial state $\ket{0} \ket{c_1}\dots\ket{c_m} \ket{0}$\;
Choose $j \in [1, m]$ uniformly at random
$m_{tot}\gets m_{tot} + 2 j + 1$\;
Apply $\left(\mathcal{G}S_0\mathcal{G}^\dagger S_F\right)^{j}$ to $\ket{s} = \mathcal{G} \ket{0} \ket{c_1}\dots\ket{c_m} \ket{0}$\;
%\Comment{$S_{0/f}$ flips phase of zero/good state}
Measure all registers $\rightarrow\ket{x}
\ket{\tilde{c}_1}\dots\ket{\tilde{c}_m}\ket{P_{x}}$\;
$V\gets \sum_{k=1:\tilde{c}_k<0}^m \tilde{c}_k$\;
\eIf{$F(V) = 1$ \textbf{or} $m_{tot} \geq M$}{
    \textbf{return} $x$, $V$\;
}{
    go to step 2\;
}
\end{algorithm}

The \autoref{conjecture:general_problems} defines a framework for a general-purpose Grover-based quantum heuristic that enhances exhaustive search by performing rudimentary intermediate constraint checks during the initial-state preparation routine to reduce branching.
This framework may be further extended by increasing the computational effort of the state preparation, but this is beyond the scope of this paper.

\subsection{Additional quantum search stage}
\autoref{alg:QSearch} and \ref{alg:QMaxSearch} describe an algorithm that always tries to improve the objective value of the current best-found solution, while always rejecting non-feasible solutions. 
This method is very successful when dealing with low-complex constraints.
For very complex constraints, this can be problematic, as a feasible solution may not be found within a reasonable time, given that no viable solution was initially provided to the algorithm.
This can be prevented if the QMaxSearch is separated into two stages.
We defining the \textit{remaining cost} as 
\begin{align}
    \tilde{c}_k = c_k + P_k - g_k(x),
\end{align}
where $g_k(x)$ is the general constrant of \autoref{eq:MostGeneralConstraint}.
If the initial solution is not feasible, compute the total constraint violation, for example
\begin{align}
    V = \sum_{j = 1:\tilde{c}_j < 0}^m \tilde{c}_k,
\end{align}

\begin{algorithm}[t]
\caption{\textbf{QMaxSearchTwoStage}($\mathcal{G}$, $M$, $F$, $c$)}\label{alg:QMaxSearchTwoStage}
set $T=\sum_{k=1:\tilde{c}_k<0}^m \tilde{c}_k$, $x=0$\;
stage $ = 1$\;
\If{x is feasible}{
    stage $=2$\;
    $T = 0$\;
}
\While{True}{
    \eIf{stage = 1}{
        $x$, $V$ = \textbf{QSearchViolation}($\mathcal{G}$, $T$, $M$)\;
        \If{$V > 0$}{
            stage $= 2$\;
            $T = P_x$, $V = -\infty$\;
        }
    }{
        $x$, $V$ = \textbf{QSearch}($\mathcal{G}$, $F(x)=x>T$, $M$)\;
    }
    \eIf{$V > T$}{
        set $T = V$ and $x_{out} = x$\;
    }{
        \textbf{return} $x_{out}$, $T$\;
    }
}
\end{algorithm}

and Grover over the violation until $V$ becomes positive.
Afterwards, the original search is executed.
This QSearch variant is described by \autoref{alg:QSearch2}.
The full two-stage quantum search is presented in \autoref{alg:QMaxSearchTwoStage}.
In principle, a three-stage search could be implemented: after a feasible solution is found (i.e., all $\tilde{c}_k$ are positive), the Grover search could be used to tighten the constraints by searching for assignments that minimise the total remaining capacity, ensuring all values remain positive.

\subsection{Benchmarking\label{sec:benchmarking}}
We aim to show the potential performance of the state preparation routine. 
As done in previous work, we use our developed quantum circuit as an initial state for Grover's maximum-finding described by \autoref{alg:QSearch} and \ref{alg:QMaxSearch}.
The first goal is to demonstrate the improvement of quantum search, specifically reducing the number of Grover iterations required to find the optimal solution compared to the standard Grover algorithm.
Secondly, we also aim to compare the performance of our quantum algorithm to classical state-of-the-art optimization solvers.
The classical solvers are Gurobi, Hexaly and Simanneal, a simulated annealing solver. 
While Gurobi and Hexaly are exact solvers, we can log the time stamps of any incumbent solution found without performance guarantees, treating those as heuristic solutions. 
We compare the time-to-solution of the quantum and classical algorithms for every incumbent solution.
For benchmarking purposes, we define an optimization problem with a quadratic objective function and two quadratic constraints
\begin{align}
\begin{aligned}
\label{eq:benchmarking_problem}
	\max &\sum_{i=1}^{n}\sum_{j=i}^{n} p_{i,j}x_{i}x_{j}\\
	s.t. 	&\sum_{i=1}^{n}\sum_{j=i}^{n} w_{i,j,1}x_{i}x_{j}\leq c_{1}\\
		&\sum_{i=1}^{n}\sum_{j=i}^{n} w_{i,j,2}x_{i}x_{j}\geq c_{2}\\
		& x_{j} \in \{0,1\} \text{ } 1\leq j \leq n.
\end{aligned}
\end{align}

Our goal is not to provide an exhaustive investigation, but rather to demonstrate the potential of our quantum algorithm.
For this, we investigate instances of sizes $n=10-3000$, with $10$ randomly generated instances per instance size.
Importantly, every instance using the classical state-of-the-art solvers has proven to be satisfiable, as we developed a heuristic incapable of establishing the unsatisfiability of any instance.
The classical solvers are given a time limit of $9$ minutes for every instance. 
The quantum algorithm is given a maximum number of oracle calls of $(n / 4)^{2} + 1200$ and a rotation angles $\theta_0 = 2\arccos \sqrt{\frac{n+4}{n+8}}$ and $\theta_1 = 2\arccos \sqrt{\frac{4}{n+8}}$.
The $R_y$ gate is used to bias bit-wise towards a given assignment. 
$\theta_{x_j}$ is therefore used for the bit assignment $x_j$.
To accurately estimate the performance of the quantum algorithm, we build upon the benchmarking methods depicted in \cite{wilkening2024quantumalgorithmsolving01, ODSQ1}.
We translate the methods from \autoref{theorem:general_objective} and \ref{theorem:constraint} into a classical routine, where a classical probabilistic decision replaces the $R_y(\theta_{x_j})$ operation.
In \cite{ODSQ1}, it was demonstrated that this sampling-based benchmarking approach yields results very close to those obtained from exact benchmarking setups on small instances.
By using this routine within a sampling algorithm, we can determine the number of Grover iterations required to find various solutions to the optimization problem.
When benchmarking our algorithm, we can log some specific data to compute the expected time for NQS to solve the problem.
\begin{figure}[t!]
    \centering
	\includegraphics[width=\linewidth]{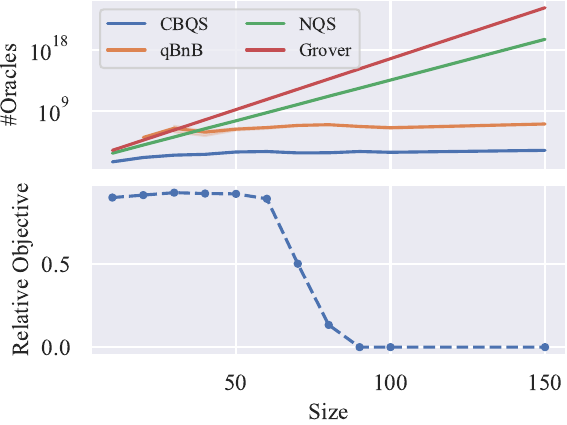}
	\caption{
    Comparison of the number of oracle calls required by Grover's algorithm, qBnB and CBQS (above plot), and the relative objective value of Simanneal compared to CBQS. 
    Simulated annealing is limited to using 1$0^5$ steps, leading to a vanishing success probability in finding feasible solutions from 90 variables.
    }
   	\label{fig:iqs-vs-grover}
\end{figure}
\begin{figure*}[t!]
    \includegraphics[width=0.9\linewidth]{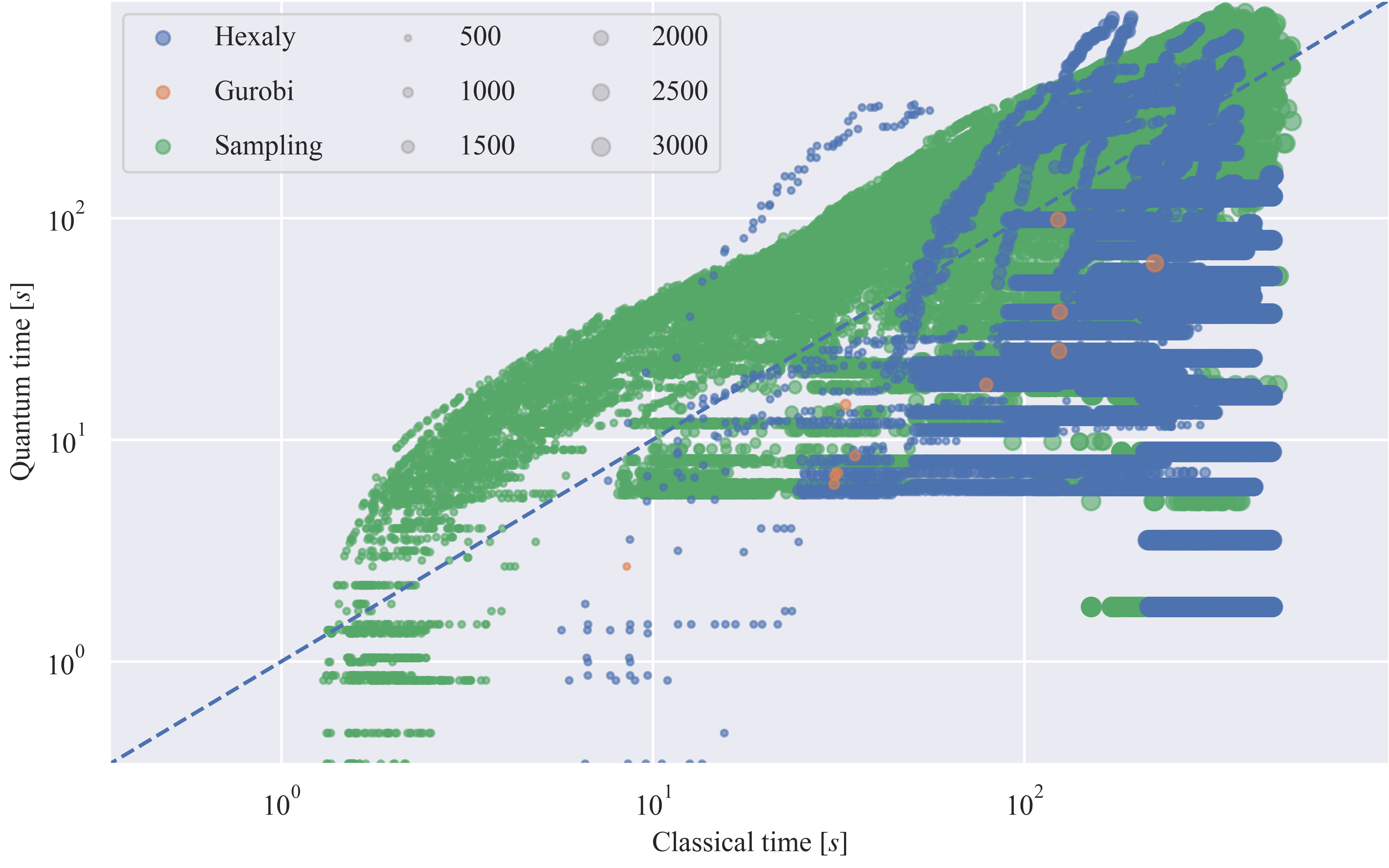}
    \caption{
	Comparison of the \textit{time-to-solution} of CBQS versus the classical state-of-the-art solvers Gurobi and Hexaly.
	The quantum gates are assumed to be executed within $6.5ns$.
	Our quantum algorithm potentially outperforms the classical solvers at finding various incumbent solutions.
    }
    \label{fig:comparison}
\end{figure*}

Using this classical sampling routine, we can determine the number of oracle calls required to find an incumbent solution.
To determine the algorithm's actual runtime, we further need to compute the exact number of quantum operations applied.
Here, we count the number of quantum cycles/layer.
At every iteration, the quantum algorithm has to apply a unitary computing and uncomputing $Q^+$ and $Q^-$ from \autoref{eq:Q_plus} and \autoref{eq:Q_minus}.
If variable $t$ is targeted, this includes every term of the matrix with indices $j$, resulting in $n$ terms. 
We assume that for all constraints, we can do the computation of the $t$ terms in parallel. 
In consecutive applications of the QFT adder, the QFTs and inverse QFTs cancel each other out.

Using the improvement from \cite{wilkening2024quantumalgorithmsolving01}, this step results in
\begin{align}
    C_A(n) = n (2 \lceil \log_2 \kappa \rceil + 3).
\end{align}
cycles.
Afterwards, a comparison with the remaining capacity is performed, as described by \autoref{eq:boolean}, and a conditional subtraction of the respective $Q^+$ and $Q^-$.
Representing the integer values in the two's complement, this can be done by subtraction and a single CNOT gate.
As this check is uncontrolled, for a large number of integer qubits $\kappa$, a different adder implementation is superior.
This is the implementation of a quantum carry-lookahead adder \cite{draper2004logarithmicdepthquantumcarrylookaheadadder} with a depth of
\begin{align}
\begin{aligned}
    C_{A2} = \lfloor\log_2 \kappa \rfloor + &\lfloor\log_2 (\kappa - 1) \rfloor + \lfloor\log_2 \frac{\kappa}{3} \rfloor + \\
    &\lfloor\log_2 \frac{\kappa - 1}{3}\rfloor + 14.
\end{aligned}
\end{align}
The full constraint check, then adds up to a total cycle count of
\begin{align}
    2n(C_{A2} + 4\kappa + 1) + 2n C_A(n).
\end{align}
The quadratic objective value is evaluated similar to $Q^+$ and $Q^-$ by consecutive QFT adders, adding a total depth of $C_A(\frac{n(n+1)}{2})$.

To compare our quantum search with qBnB, we focus solely on the number of queries the algorithm makes, favouring qBnB. 
This is a generous comparison for qBnB, as it must employ costly quantum routines to compute inner and outer bounds while exploring the tree. 
Modern classical Branch-and-Bound solvers, such as Gurobi and Hexaly, utilise various heuristics, including the simplex algorithm \cite{huang2021branchboundmixedinteger}, to compute these bounds, which are essential for translating into a quantum routine for a general-purpose qBnB solver.
Although this is possible due to \cite{nannicini2022fastquantumsubroutinessimplex}, recent work suggests that this approach is practically unusable \cite{ammann2023realisticruntimeanalysisquantum}.
Nevertheless, for the mentioned comparison of the quantum solvers, we execute Gurobi, extracting the total number of nodes $T$ explored by the search strategy, resulting in the number of oracle calls for qBnB given by $\order{\sqrt{nT}n\log_{2}^{2}(n\log_{2}(T))}$ \cite{Ambainis_2017}. 
Here, we ignore the accuracy factor.

Since we can only compute the theoretical runtime of our quantum algorithm, we must make certain assumptions about the quantum hardware. 
As we want to demonstrate the potential for a far-term quantum advantage, we make some optimistic but not unrealistic assumptions.
We assume our quantum hardware can execute a single controlled $R_{y}$ gate with an arbitrary angle $\theta$, the Toffoli gate, and any single-qubit logic gate.
Further, we assume that all disjoint quantum gates can be executed in parallel, culminating in a single \textit{quantum cycle}, which we consider to be executable in a time of $6.5ns$ \cite{Chew2022UltrafastEnergyExchangeBetweenTwoSingleRydbergAtomsOnANanosecondTimescale}.
All the experiments were performed on a MacBook Pro with an (6 (12 Threads) x 2.6-4.1 GHz) Intel Core i7. 
The code, experiments, and data can be found in \cite{Software_code, Software}. 

\section{Results}

Initial benchmarking was conducted on minor problem instances (below 150 variables). As illustrated in \autoref{fig:iqs-vs-grover}, basic heuristics such as simulated annealing (SA) struggle to find even a feasible solution when the instance size exceeds 90 variables. Moreover, SA exhibits longer runtimes compared to our proposed quantum routine. Due to its poor performance and inability to produce feasible solutions for larger instances, SA was excluded from those benchmarks.

\autoref{fig:iqs-vs-grover} also compares the number of oracle calls required by various quantum routines. Notably, our quantum search framework demonstrates a clear improvement over Grover’s algorithm and NQS. Since the oracles involved are of comparable cost, the observed enhancement in oracle efficiency translates directly into a corresponding reduction in quantum gate usage. Furthermore, our method consistently requires fewer oracle calls than the qBnB method. 
For qBnB, the number of nodes explored by Gurobi was used as a proxy, as equivalent data were not available from Hexaly. 
Even after disregarding the computationally intensive quantum simplex subroutine, our method consistently outperforms qBnB in terms of oracle efficiency as the problem size increases. 
Based on these numerical results, our approach appears to be the most promising (at least far-term) quantum algorithm for addressing the considered optimization problem, potentially also for general optimization problems.

After empirically demonstrating that our algorithm performs best among fault-tolerant quantum algorithms, we compare its \textit{time-to-solution} with state-of-the-art classical solvers. 
Firstly, we observe in \autoref{fig:comparison} that Hexaly seems to outperform Gurobi on many instances. 
Yet, assuming we have access to a powerful quantum computer, we observe that our proposed quantum algorithm can potentially outperform the classical solvers. 
For many data points, it also appears possible to achieve quantum advantage even with slower quantum gates and error-correction schemes, given the large runtime gap of factors of up to $10^3$.
As the quantum algorithm is benchmarked using a specialized sampling routine, we can also log its time-to-solution. 
As expected, for many, quickly finding assignments (via a few Grover iterations) makes the sampling algorithm faster, as it requires drawing only a small number of samples. 
But for most data points, the quantum algorithm will outperform its classical equivalent.
Nevertheless, the results suggest that the quantum-inspired classical sampling algorithm may be a promising candidate for classical optimization solvers.
This might be even more emphasized by the investigation in \autoref{fig:obj-over-time}.
Here, all the algorithms are executed on a 1000-variable instance with a time limit of 25000 seconds ($\approx$7 hours). 
For the quantum algorithm, the sampling routine had to run longer.
The plot demonstrates that our classical sampling and its corresponding quantum algorithm are much more performant on the investigated problem, as they don't struggle to find increasingly better incumbent solutions.
Here, it is not known how close the best-found solution (CBQS) is to the optimal solution, as the classical algorithms couldn't achieve an optimality gap better than $50\%$.

\begin{figure*}[t]
	\centering
	\includegraphics[width=\textwidth]{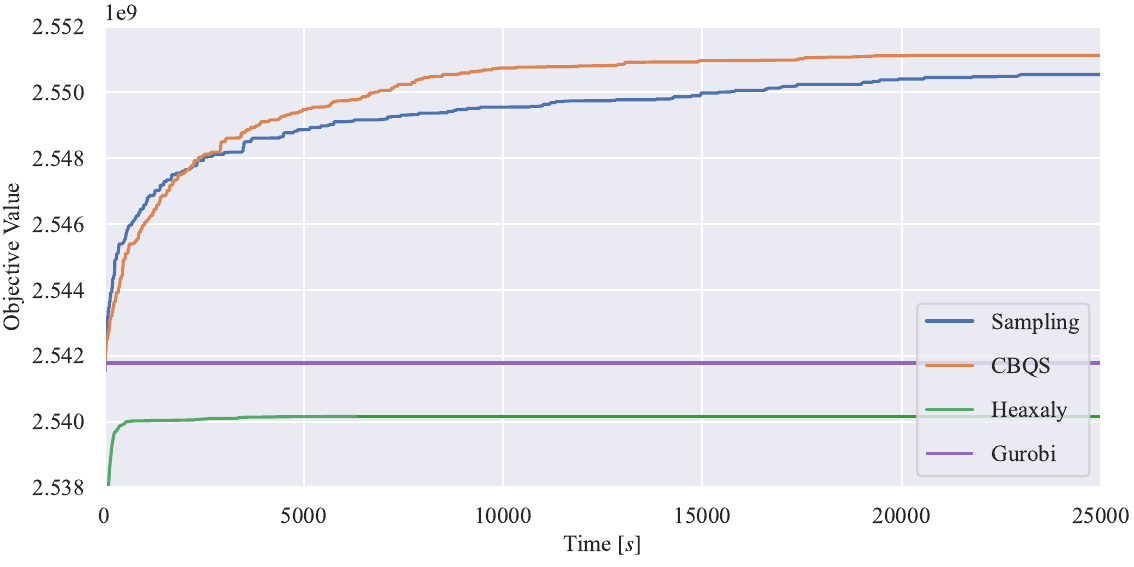}
	\caption{Best found incumbent solution over time up to $\approx$7h for one 1000 variable instance. The comparison was made between Hexaly, Gurobi, CBQS, and its classical sampling equivalent.
	The quantum algorithm and its classical counterpart can quickly outperform Hexaly and Gurobi, which struggle to improve upon the incumbent solution by much.
	}
    \label{fig:obj-over-time}
\end{figure*}

\section{Discussion}
In this paper, we present a Grover-based quantum heuristic for general optimization problems and demonstrate its potential superiority over classical state-of-the-art methods.
Here, this was demonstrated on a specific problem with a large benchmark set, yet we fail to provide predictions for various optimization problems.
Still, a more rigorous benchmark analysis or theoretical development is required to affirm our findings.
Nevertheless, the success of this and the previous work leaves us optimistic about the capabilities of our routine.
Furthermore, our algorithm, as it currently stands, is a relatively primitive implementation.
We believe numerous improvements can be made to this framework, thereby enhancing the algorithms' overall performance and efficiency.
Those improvements may also be unavoidable when more realistic hardware assumptions must be considered, which increases circuit depth to support error correction and higher gate times.
Furthermore, as the presented quantum algorithm is a form of local search, extending it with ideas from the general metaheuristic approach, tabu search \cite{glover1990tabu}, might further enhance its performance.

An interesting future work would also be to provide a rigorous proof for \autoref{conjecture:general_problems} and a general algorithmic description for this class of quantum algorithms, beyond the provided formulations.

\section*{Acknowledgments}

We thank René Schwonneck, Lennart Binksowski, Timo Ziegler, Maximilian Hess, and Tobias J. Osborne for constructive discussions.
This project was enabled by the DFG through SFB 1227(DQ-mat), QuantumFrontiers, the QuantumValley Lower Saxony, the BMBF projects ATIQ and QuBRA, the BMWK project ProvideQ, and the Quantera project ResourceQ.

\bibliographystyle{unsrt}
\bibliography{sample}

\begin{thebibliography}{10}

\bibitem{Durr1996AQuantumAlgorithmForFindingTheMinimum}
Christoph Durr and Peter H{\o}yer.
\newblock {A Quantum Algorithm for Finding the Minimum}.
\newblock [arXiv preprint
  \href{https://arxiv.org/abs/quant-ph/9607014}{arXiv:quant-ph/9607014}], 1996.

\bibitem{Grover1996AFastQuantumMechanicalAlgorithmForDatabaseSearch}
Lov~K. Grover.
\newblock {A fast quantum mechanical algorithm for database search}.
\newblock [arXiv preprint
  \href{https://arxiv.org/abs/quant-ph/9605043}{arXiv:quant-ph/9605043}], 1996.

\bibitem{Nemhauser1988IntegerAndCombinatorialOptimization}
George Nemhauser and Laurence Wolsey.
\newblock {\em The Scope of Integer and Combinatorial Optimization}, chapter
  I.1, pages 1--26.
\newblock John Wiley and Sons, Ltd, 1988.

\bibitem{Karp1975NPHardness}
R.~M. Karp.
\newblock On the computational complexity of combinatorial problems.
\newblock {\em Networks}, 5(1):45--68, 1975.

\bibitem{Stoudenmire_2024}
E.M. Stoudenmire and Xavier Waintal.
\newblock Opening the black box inside grover’s algorithm.
\newblock {\em Physical Review X}, 14(4), November 2024.

\bibitem{Aaronson2023blog}
Scott Aaronson.
\newblock Of course grover’s algorithm offers a quantum advantage!
\newblock \url{https://scottaaronson.blog/?p=7143}.

\bibitem{wilkening2024quantumalgorithmsolving01}
S{\"o}ren Wilkening, Andreea-Iulia Lefterovici, Lennart Binkowski, Michael
  Perk, S{\'a}ndor~P. Fekete, and Tobias~J. Osborne.
\newblock A quantum algorithm for solving 0-1 knapsack problems.
\newblock {\em npj Quantum Information}, 11(1):146, 2025.

\bibitem{wilkening2025quantumsearchmethodquadratic}
Sören Wilkening, Andreea-Iulia Lefterovici, Lennart Binkowski, Marlene Funck,
  Michael Perk, Robert Karimov, Sándor Fekete, and Tobias~J. Osborne.
\newblock A quantum search method for quadratic and multidimensional knapsack
  problems.
\newblock In {\em 2025 IEEE International Conference on Quantum Computing and
  Engineering (QCE)}, volume~01, pages 2209--2215, 2025.

\bibitem{ODSQ1}
Sören Wilkening, Timo Ziegler, and Maximilian Hess.
\newblock Constraint-oriented biased quantum search for linear constrained
  combinatorial optimization problems, 2025.

\bibitem{gurobi}
{Gurobi Optimization, LLC}.
\newblock {Gurobi Optimizer Reference Manual}, 2024.

\bibitem{Leeuwen2012TheoryOfComputingSystems}
Erik~Jan van Leeuwen and Jan van Leeuwen.
\newblock Structure of polynomial-time approximation.
\newblock {\em Theory of Computing Systems}, 50(4):641--674, 2012.

\bibitem{Deb2014}
Kalyanmoy Deb and Kalyanmoy Deb.
\newblock {\em Multi-objective Optimization}, pages 403--449.
\newblock Springer US, Boston, MA, 2014.

\bibitem{hexaly}
Hexaly.
\newblock Hexaly optimizer, 2024.
\newblock Version 13.5.

\bibitem{simanneal}
Simanneal.
\newblock \url{https://github.com/perrygeo/simanneal}.

\bibitem{Montanaro_2020}
Ashley Montanaro.
\newblock Quantum speedup of branch-and-bound algorithms.
\newblock {\em Physical Review Research}, 2(1), January 2020.

\bibitem{chakrabarti2022universalquantumspeedupbranchandbound}
Shouvanik Chakrabarti, Pierre Minssen, Romina Yalovetzky, and Marco Pistoia.
\newblock Universal quantum speedup for branch-and-bound, branch-and-cut, and
  tree-search algorithms, 2022.

\bibitem{Ambainis_2017}
Andris Ambainis and Martins Kokainis.
\newblock Quantum algorithm for tree size estimation, with applications to
  backtracking and 2-player games.
\newblock In {\em Proceedings of the 49th Annual ACM SIGACT Symposium on Theory
  of Computing}, STOC ’17. ACM, June 2017.

\bibitem{yue2023localglobaldistributedquantum}
Bo~Yue, Shibei Xue, Yu~Pan, Min Jiang, and Daoyi Dong.
\newblock Local to global: A distributed quantum approximate optimization
  algorithm for pseudo-boolean optimization problems, 2023.

\bibitem{Lykov_2023}
Danylo Lykov, Ruslan Shaydulin, Yue Sun, Yuri Alexeev, and Marco Pistoia.
\newblock Fast simulation of high-depth qaoa circuits.
\newblock In {\em Proceedings of the SC ’23 Workshops of the International
  Conference on High Performance Computing, Network, Storage, and Analysis},
  SC-W 2023, page 1443–1451. ACM, November 2023.

\bibitem{dubey2025simulatingquantumcircuitstree}
Aditya Dubey, Zeki Zeybek, and Peter Schmelcher.
\newblock Simulating quantum circuits with tree tensor networks using
  density-matrix renormalization group algorithm, 2025.

\bibitem{christiansen2024quantumtreegeneratorimproves}
Paul Christiansen, Lennart Binkowski, Debora Ramacciotti, and Sören Wilkening.
\newblock Quantum tree generator improves qaoa state-of-the-art for the
  knapsack problem, 2024.

\bibitem{Cerf2000Nested!uantumSearchAndStructuredProblems}
Nicolas~J. Cerf, Lov~K. Grover, and Colin~P. Williams.
\newblock Nested quantum search and structured problems.
\newblock {\em Physical Review A}, 61(3), February 2000.

\bibitem{huang2021branchboundmixedinteger}
Lingying Huang, Xiaomeng Chen, Wei Huo, Jiazheng Wang, Fan Zhang, Bo~Bai, and
  Ling Shi.
\newblock Branch and bound in mixed integer linear programming problems: A
  survey of techniques and trends, 2021.

\bibitem{nannicini2022fastquantumsubroutinessimplex}
Giacomo Nannicini.
\newblock Fast quantum subroutines for the simplex method, 2022.

\bibitem{ammann2023realisticruntimeanalysisquantum}
Sabrina Ammann, Maximilian Hess, Debora Ramacciotti, Sándor~P. Fekete, Paulina
  L.~A. Goedicke, David Gross, Andreea Lefterovici, Tobias~J. Osborne, Michael
  Perk, Antonio Rotundo, S.~E. Skelton, Sebastian Stiller, and Timo de~Wolff.
\newblock Realistic runtime analysis for quantum simplex computation, 2023.

\bibitem{Brassard_2002}
Gilles Brassard, Peter Høyer, Michele Mosca, and Alain Tapp.
\newblock Quantum amplitude amplification and estimation, 2002.

\bibitem{Bernstein1997QuantumComplexityTheory}
Ethan Bernstein and Umesh Vazirani.
\newblock Quantum complexity theory.
\newblock {\em SIAM Journal on Computing}, 26(5):1411--1473, 1997.

\bibitem{Boros2002PseudoBooleanOptimizations}
Endre Boros and Peter~L. Hammer.
\newblock Pseudo-boolean optimization.
\newblock {\em Discrete Applied Mathematics}, 123(1):155--225, 2002.

\bibitem{pferschy2021productknapsackproblemapproximation}
Ulrich Pferschy, Joachim Schauer, and Clemens Thielen.
\newblock The product knapsack problem: Approximation and complexity, 2021.

\bibitem{Junhong2023QuantumCircuitDesignForIntegerMultiplicationBasedOnSchönhageStrassenAlgorithm}
Junhong Nie, Qinlin Zhu, Meng Li, and Xiaoming Sun.
\newblock Quantum circuit design for integer multiplication based on
  schönhage–strassen algorithm.
\newblock {\em IEEE Transactions on Computer-Aided Design of Integrated
  Circuits and Systems}, 42(12):4791--4802, 2023.

\bibitem{draper2004logarithmicdepthquantumcarrylookaheadadder}
Thomas~G. Draper, Samuel~A. Kutin, Eric~M. Rains, and Krysta~M. Svore.
\newblock A logarithmic-depth quantum carry-lookahead adder, 2004.

\bibitem{Chew2022UltrafastEnergyExchangeBetweenTwoSingleRydbergAtomsOnANanosecondTimescale}
Y.~Chew, T.~Tomita, T.~P. Mahesh, S.~Sugawa, S.~de~L{\'e}s{\'e}leuc, and
  K.~Ohmori.
\newblock Ultrafast energy exchange between two single rydberg atoms on a
  nanosecond timescale.
\newblock {\em Nat. Photonics}, 16(10):724--729, 2022.

\bibitem{Software_code}
S\"oren Wilkening.
\newblock constraint-oriented-biased-quantum-search.
\newblock
  \url{https://github.com/SoerenWilkening/constraint-oriented-biased-quantum-search}.

\bibitem{Software}
S\"oren Wilkening.
\newblock Cbqs-benchmarks.
\newblock \url{https://github.com/SoerenWilkening/CBQS-benchmarks}.

\bibitem{glover1990tabu}
Fred Glover.
\newblock Tabu search: A tutorial.
\newblock {\em Interfaces}, 20(4):74--94, 1990.

\end{thebibliography}

\end{document}